\documentclass[12pt,reqno,twoside]{amsart}

\usepackage{amsthm,amsmath,amssymb,mathrsfs,amsbsy}
\usepackage{graphicx,epsfig,wrapfig,sidecap,subfig}
\usepackage[all]{xy}
\usepackage{proof}
\usepackage[usenames]{color}
\usepackage{hyperref}
\usepackage{bm}
\usepackage[capitalize]{cleveref}

\xyoption{dvips}
\xyoption{color}

\usepackage{color}

\title{A Danzer set for Axis Parallel Boxes}
\date{}
\author{David Simmons}
\address{University of York, Department of Mathematics, Heslington, York YO10 5DD, UK}
\email{David.Simmons@york.ac.uk}
\urladdr{\url{https://sites.google.com/site/davidsimmonsmath/}}
\author{Yaar Solomon}
\address{Stony Brook University, Department of Mathematics, Stony Brook, NY}
\email{yaar.solomon@stonybrook.edu}
\urladdr{\url{http://www.math.stonybrook.edu/~yaars/}}

\newcommand{\Rplus}{[0,\infty)}

\newcommand{\N}{{\mathbb{N}}}
\newcommand{\Z}{{\mathbb{Z}}}

\newcommand{\R}{{\mathbb{R}}}

\newcommand{\RR}{{\mathcal{R}}}
\newcommand{\OO}{\mathcal{O}}
\newcommand{\ttt}{\mathbf{t}}
\newcommand{\xx}{\mathbf{x}}
\newcommand{\yy}{\mathbf{y}}

\newcommand{\df}{{\, \stackrel{\mathrm{def}}{=}\, }}
\newcommand{\minus}{\smallsetminus}
\newcommand{\Fin}{\{0,1\}^\Z_{Fin}}

\newcommand{\Finoverk}{\{0,1\}^{\ge k}_{Fin}}
\newcommand{\Finlessk}{\{0,1\}^{< k}_{Fin}}

\newcommand{\absolute}[1] {\left|{#1}\right|}

\newcommand{\norm}[1]{\left\|{#1}\right\|}

\newcommand{\Vol}{\mathrm{Vol}}
\DeclareMathOperator{\SL}{SL}

\newcommand {\ignore}[1]  {}

\theoremstyle{plain}
\newtheorem{thm}{Theorem}[section]

\newtheorem{cor}[thm]{Corollary}

\theoremstyle{definition}
\newtheorem{definition}[thm]{Definition}

\newtheorem{remark}[thm]{Remark}

\numberwithin{equation}{section}
\swapnumbers

\begin{document}
\begin{abstract}
We present concrete constructions of discrete sets in $\R^d$ ($d\ge 2$) that intersect every aligned box of volume $1$ in $\R^d$, and which have optimal growth rate $O(T^d)$.
\end{abstract}
\maketitle
\section{Introduction}\label{sec:introduction}
A set $D\subseteq\R^d$ is called a \emph{Danzer set} if there exists an $s>0$ such that $D$ intersects every convex set of volume $s$. The question whether a discrete Danzer set in $\R^d$ of growth rate $O(T^d)$ exists is due to Danzer, see \cite{CFG, Gowers, GL}, and has been open since the sixties.

There are several variants of this question. One is to weaken the Danzer property in the following sense. We say that $Y\subseteq\R^d$ is a \emph{dense forest} if there is a function $\varepsilon=\varepsilon(T)\xrightarrow{T\to\infty}0$ so that for every $x\in\R^d$ and for every direction $v\in\mathcal{S}^{d-1}$, the distance between $Y$ and the line segment of length $T$ which starts at $x$ and proceeds in direction $v$ is less than $\varepsilon(T)$. Intuitively, as it was presented in \cite{Bishop}, $T$ is the maximal distance that a man can see when standing in a forest with a trunk of radius $\varepsilon$ located at each element of $Y$. Note that every Danzer set is a dense forest with $\varepsilon(T)=O(T^{-1/(d - 1)})$, and a dense forest with $\varepsilon(T)=O(T^{-(d - 1)})$ is a Danzer set.\footnote{The second statement is proven as follows: let $D$ be a dense forest with $\varepsilon(T) = O(T^{-(d - 1)})$, and let $R\subseteq\R^d$ be a box (i.e. a parallelotope with adjacent faces orthogonal) with volume $s$ and shortest edge length $2\varepsilon$. Since the volume of a box is the product of the length of its sides, $R$ has an edge of length at least $T:= \left(\frac{s}{2\varepsilon}\right)^{1/(d-1)}$. Let $L$ be the line segment parallel to this edge, passing through the center of $R$, and of length $T-2\varepsilon$. If $R$ does not contain any points of $D$, then the distance from $L$ to $D$ is at least $\varepsilon$, which implies that $\varepsilon\leq O(T^{-(d - 1)}) = O(\varepsilon/s)$. For $s$ sufficiently large, this is a contradiction, so every box of sufficiently large volume intersects $D$. Since every convex set contains a box of volume at least a constant times the volume of the convex set, this shows that $D$ is a Danzer set.} A construction of a dense forest of growth rate $O(T^d)$ is given in \cite{SW}, and another construction in the plane follows from the proof of \cite[Lemma 2.4]{Bishop}.

One other interesting direction is to look for Danzer sets with faster growth rates. A Danzer set of growth rate $O(T^d(\log T)^{d-1})$ is given in \cite{BW}; this bound was improved recently in \cite{SW} by a probabilistic construction that gives growth rate $O(T^d\log T)$.

Another approach in trying to weaken the Danzer problem is by hitting a smaller family of sets, instead of all the convex sets. John's theorem \cite{John} implies that replacing convex sets by boxes\footnote{A \emph{box} in $\R^d$ is the image of an aligned box $[a_1,b_1]\times\cdots\times[a_d,b_d]$ under an orthogonal matrix.} gives an equivalent question. In this note we consider a question that arises naturally from the Danzer problem. We say that $D\subseteq\R^d$ is an \emph{align-Danzer set} if there is an $s>0$ such that $D$ intersects every aligned box of volume $s$. In our main results, Theorem \ref{thm:main1} and Theorem \ref{thm:main2} below, we present simple constructions for align-Danzer sets in $\R^d$ of growth rate $O(T^d)$. Neither of these constructions is new, but the viewpoint of seeing them as connected with Danzer's problem is new.

We denote by $\Fin$ the subset of $\{0,1\}^\Z$ consisting of those bi-infinite sequences that contain only finitely many $1$s.

\begin{thm}\label{thm:main1}
The set 
\[
D\df \left\{\left(\pm\sum_{n\in\Z}a_n 2^n,\pm\sum_{n\in\Z}a_n 2^{-n}\right)\in\R^2:(a_n)\in\Fin\right\}
\]
is an align-Danzer set in $\R^2$ of growth rate $O(T^2)$.
\end{thm}

The set in Theorem \ref{thm:main1} is a variant of the binary version of the well-known van der Corput sequence (see e.g. \cite{vdC}).

Although the set $D$ in Theorem \ref{thm:main1} is given very explicitly, and the proof is by elementary means, it only solves the problem in dimension $2$, and no simple higher-dimensional extension comes to mind. To solve the problem in higher dimensions we use a dynamical approach. 

For a fixed $d\ge 2$ let $A\subseteq \SL_d(\R)$ be the subgroup of diagonal matrices with positive entries, and let $\Omega$ be the space of all lattices in $\R^d$. 
\begin{definition}[{\cite[p.6]{Skriganov}}]\label{def:admissible}
A lattice $\Lambda\in\Omega$ is \emph{admissible} if its orbit under $A$ is precompact in $\Omega$.
\end{definition}


\begin{thm}[Corollary of {\cite[Theorem 1.2]{Skriganov}}]\label{thm:main2}
For every $d\ge 2$ there exists an admissible lattice in $\R^d$, and every admissible lattice is an align-Danzer set.
\end{thm}
Although Theorem \ref{thm:main2} is a direct consequence of \cite[Theorem 1.2]{Skriganov}, we provide the proof since it is elementary. We also refer to the discussions in \cite[p. 24-31]{GL} for additional reading. 

As a direct consequence we reprove a result in computational geometry, that follows from a result of Halton on low discrepancy sequences, see \cite{Ha}. We remark that Corollary \ref{cor:computational_geometry} is not stated in \cite{Ha}, but it is well known in the computational geometry and combinatorics communities that Halton's construction satisfies it. 
\begin{cor}\label{cor:computational_geometry}
For every $\varepsilon>0$ there are $\varepsilon$-nets of optimal sizes $O(1/\varepsilon)$ for the range space $(X,\RR)$, where $X=[0,1]^d$ and $\RR=\{\text{aligned boxes}\}$.
\end{cor}
This Corollary follows directly from the above Theorems by restricting to a bounded cube and rescaling to $[0,1]^d$. We refer to \cite{AS,M} for a more comprehensive reading about the notions in Corollary \ref{cor:computational_geometry}.

\begin{remark}
Align Danzer sets in $\R^d$ of growth rate $O(T^d)$ can also be constructed by modifying the proof of \cite[Theorem 1.4]{SW} to work for aligned boxes and then combining with the result of \cite{Ha} or \cite{vdC} in the unit cube. Nonetheless, our constructions here are simple and the proofs are straightforward. 
\end{remark}

\subsection{Acknowledgements}
We thank Sathish Govindarajan, Shakhar Smorodinsky, and Barak Weiss for useful discussions that helped us understand the status of the problem. We also thank the referee for helpful comments.

\section{Proof of Theorem \ref{thm:main1}}
\begin{proof}[Proof of Theorem \ref{thm:main1}]
We first show that $D$ intersects every aligned box of volume $64$. It suffices to show that 
\[D_+\df \left\{\left(\sum_{n\in\Z}a_n2^n,\sum_{n\in\Z}a_n2^{-n}\right)\in\R^2:(a_n)\in\Fin\right\}\]
intersects every aligned box of volume $16$ that sits in $\R^2_+\df \Rplus^2$. 

Let $R\subseteq\R^2_+$ be an aligned box of volume $16$, and denote its lower left vertex by $(x,y)$. Let $t>0$ be such that the lower right and the upper left vertices of $R$ are $(x+t,y)$ and $(x,y+\frac{16}{t})$ respectively. We define a sequence $(a_n)_{n\in\Z}\in\Fin$ so that $\left(\sum_{n\in\Z}a_n2^n,\sum_{n\in\Z}a_n2^{-n}\right)\in R$.

For each integer $k$, we denote by $\Finoverk$ and $\Finlessk$ the subsets of $\{0,1\}^{\geq k}$ and $\{0,1\}^{< k}$, respectively, consisting of those sequences that contain only finitely many $1$s. Here $\{0,1\}^{\geq k}$ is the set of all sequences in $\{0,1\}$ of the form $(a_k,a_{k + 1},\ldots)$, and $\{0,1\}^{< k}$ is the set of all sequences in $\{0,1\}$ of the form $(\ldots, a_{k - 2},a_{k - 1})$.

Let $k\in\Z$ be such that $2^k\le \frac{t}{2}<2^{k+1}$. Observe that $\sum_{n<k}a_n2^n<2^k\le\frac{t}{2}$ for any sequence $(a_n)$ in $\Finlessk$, and that the interval $(x,x+\frac{t}{2})$ intersects the set 
\[2^k\N=\left\{\sum_{n\ge k} a_n2^n:(a_n)\in\Finoverk\right\}.\]
Then we may choose the $a_n$s for $n\ge k$ so that $\sum_{n\ge k} a_n2^n\in(x,x+\frac{t}{2})$, and thus for any choice of the $a_n$s for $n<k$ (and in particular for the choice described below) we have $\sum_{n\in\Z}a_n2^n\in(x,x+t)$.

The analysis of the $y$ coordinate is similar. Here $2^{-k-1}<\frac{2}{t}\le 2^{-k}$, and therefore $2^{-k+1}<\frac{8}{t}\le 2^{-k+2}$. We have $\sum_{n\ge k}a_n2^{-n}<2^{-k+1}<\frac{8}{t}$ for any sequence $(a_n)$ in $\Finoverk$, and the interval $(y,y+\frac{8}{t})$ intersects the set 
\[2^{-k+1}\N=\left\{\sum_{n<k} a_n2^{-n}:(a_n)\in\Finlessk\right\}.\]
Then we may choose the $a_n$s for $n<k$ so that $\sum_{n<k}a_n2^{-n}\in(y,y+\frac{8}{t})$, and thus for any choice of the $a_n$s for $n\ge k$ (and in particular for the choice described above) we have $\sum_{n\in\Z}a_n2^{-n}\in(y,y+\frac{16}{t})$.

It is left to show that $D$ (or $D_+$) is of growth rate $O(T^2)$. To see that, consider the set 
\[B\df\left\{\left(\sum_{n\ge 0} a_n2^n,\sum_{n<0}a_n2^{-n}\right)\in\R^2:(a_n)\in\Fin\right\}.\]
Observe that the mapping $g:D_+\to B$ which is defined in the obvious way by
\[\left(\sum_{n\in\Z}a_n2^n,\sum_{n\in\Z}a_n2^{-n}\right)\stackrel{g}\mapsto\left(\sum_{n\ge 0} a_n2^n,\sum_{n<0}a_n2^{-n}\right)\]
is a bijection, and for any $(x,y)\in D_+$ we have $\norm{(x,y)-g(x,y)}_2\le\sqrt{5}$ (where $\norm{\cdot}_2$ denotes the Euclidean norm). But since $B=\N\times 2\N$, the assertion follows.
\end{proof}

\begin{remark}
We want to stress that $D$ is not a Danzer set in $\R^2$ and not even a dense forest. To see it, observe that symmetric sequences $(a_n)$ correspond to points on the line $y=x$. On the other hand, non-symmetric sequences correspond to points $(x,y)$ with $\absolute{x-y}>1$, and in particular $D$ misses a neighborhood of the line $y=x+\frac{1}{4}$.    
\end{remark}

\section{Proof of Theorem \ref{thm:main2}}
Fix $d\geq 2$. Let $V = \{\ttt\in \R^d : \sum_{i=1}^d t_i = 0\}$, and for each $\ttt\in V$ let $g_{\ttt}\in \SL_d(\R)$ be the diagonal matrix whose entries are $e^{t_i}$. Then $\ttt\mapsto g_{\ttt}$ is a homomorphism.


\begin{proof}[Proof of Theorem \ref{thm:main2}]
Let $K$ be a totally real number field of degree $d$, and let $\OO_K$ be its ring of integers. Let $\phi_1,\ldots,\phi_d:K\to\R$ be the Galois embeddings of $K$ into $\R$, and let $\Phi: K \to \R^d$ be their direct sum. Then $\Lambda \df \Phi(\OO_K)$ is a lattice in $\R^d$. To see that $\Lambda$ is admissible, fix $\xx = \Phi(\alpha)\in\Lambda$, and observe that if $\xx\neq 0$,
\[
\prod_{i = 1}^d |x_i| = \prod_{i = 1}^d |\phi_i(\alpha)| = |N(\alpha)| \in\Z\minus\{0\}.
\]
Here $N$ denotes the norm in the field $K$. In particular, $\prod_{i = 1}^d |x_i| \geq 1$ and thus $\prod_{i = 1}^d |e^{t_i} x_i| \geq 1$ for all $\ttt\in V$. It follows that $|e^{t_i} x_i| \geq 1$ for some $i = 1,\ldots,d$ and thus $\|g_\ttt \xx\| \geq 1$. Since $\ttt$, $\xx$ were arbitrary, Mahler's compactness criterion shows that $\Lambda$ is admissible.

For the second part of the proof, let $\Lambda$ be an admissible lattice in $\R^d$. Let $R$ be an aligned box disjoint from $\Lambda$. Then there exists $\ttt\in V$ such that $g_\ttt R$ is a cube. By assumption $g_\ttt \Lambda$ is in a compact subset $K\subseteq\Omega$, hence the codiameter\footnote{The \emph{codiameter} of a lattice $\Gamma\subseteq\R^d$ is the diameter of the quotient space $\R^d/\Gamma$ (with respect to the quotient metric $d([\xx],[\yy]) = \min\{\|\yy - \xx\| : \xx,\yy \text{ representatives of }[\xx],[\yy]\}$), or equivalently the maximum of the function $\R^d\ni\xx\mapsto d(\xx,\Gamma)$. The codiameter is continuous as a function of the lattice.} of $g_\ttt\Lambda$ is bounded above by a constant independent of $\ttt$. But since $g_\ttt R$ is disjoint from $g_\ttt\Lambda$, the distance from the center of $g_\ttt R$ to the complement of $g_\ttt R$, i.e. half the edge length of the cube $g_\ttt R$, is bounded above by the distance from the center of $g_\ttt R$ to $g_\ttt\Lambda$, which is in turn bounded above by the codiameter of $g_\ttt\Lambda$. Thus both the diameter and the volume of $g_\ttt R$ are bounded above by a constant independent of $\ttt$. Since $\Vol(R) = \Vol(g_\ttt R)$, the proof is complete.
\end{proof}

\ignore{
\begin{remark}
The proof uses a dynamical characterization for a lattice $\Lambda$ to be an align-Danzer set, in terms of the orbit of $\Lambda$ under the diagonal group $A$. A similar criteria for Danzer sets appears in \cite[Proposition 3.1]{SW}, where the orbit under the whole group $\SL_d(\R)$ is considered. One may want to apply this approach with some intermediate group $A\subsetneq G\subsetneq \SL_d(\R)$, trying to construct a set that intersects a larger class of rectangles. But this approach will fail since if $A\subsetneq G$, $G$ must contain a unipotent element, which in turns implies that $G=\SL_d(\R)$.
\end{remark}

\begin{remark}
Littlewood's conjecture can be formulated in terms of admissible lattices. Specifically, let us call a lattice \emph{forward-admissible} if its orbit under the semigroup $\{g_\ttt : \ttt\in V ,\; t_i > 0 \forall i \neq 0\}$ is precompact in $\Omega$. Then Littlewood's conjecture is equivalent to the statement that no lattice of the form
\[
\begin{pmatrix}
1 &&\\
\alpha & 1 &\\
\beta && 1
\end{pmatrix}\cdot \Z^3 \qquad (\alpha,\beta\in \R)
\]
is forward-admissible.
\end{remark}
}

\end{document}